\documentclass[12pt]{article}

\usepackage{a4}
\usepackage{cite}
\usepackage{booktabs}

\usepackage{amsmath}
\allowdisplaybreaks[1]

\usepackage{amsthm}

\newtheorem{theorem}{Theorem}
\newtheorem{lemma}[theorem]{Lemma}
\newtheorem{corollary}[theorem]{Corollary}

\newcommand{\depthb}{\mathrm{depth}}
\newcommand{\parentb}{\mathrm{parent}}
\newcommand{\degreeb}{\mathrm{degree}}
\newcommand{\childrankb}{\mathrm{child\_rank}}
\newcommand{\childselectb}{\mathrm{child\_select}}
\newcommand{\numdescendantsb}{\mathrm{num\_descendants}}
\newcommand{\firstchildb}{\mathrm{first\_child}}
\newcommand{\lastchildb}{\mathrm{last\_child}}
\newcommand{\nextsiblingb}{\mathrm{next\_sibling}}
\newcommand{\prevsiblingb}{\mathrm{prev\_sibling}}
\newcommand{\heightb}{\mathrm{height}}
\newcommand{\lcab}{\mathrm{lca}}
\newcommand{\levelancestorb}{\mathrm{level\_ancestor}}
\newcommand{\levelnextb}{\mathrm{level\_next}}
\newcommand{\levelprevb}{\mathrm{level\_prev}}
\newcommand{\levellmostb}{\mathrm{level\_lmost}}
\newcommand{\levelrmostb}{\mathrm{level\_rmost}}
\newcommand{\encloseb}{\mathrm{enclose}}

\newcommand{\depth}[1]{\depthb(#1)}
\newcommand{\parent}[1]{\parentb(#1)}
\newcommand{\degree}[1]{\degreeb(#1)}
\newcommand{\childrank}[1]{\childrankb(#1)}
\newcommand{\childselect}[2]{\childselectb(#1,#2)}
\newcommand{\numdescendants}[1]{\numdescendantsb(#1)}
\newcommand{\firstchild}[1]{\firstchildb(#1)}
\newcommand{\lastchild}[1]{\lastchildb(#1)}
\newcommand{\nextsibling}[1]{\nextsiblingb(#1)}
\newcommand{\prevsibling}[1]{\prevsiblingb(#1)}
\newcommand{\height}[1]{\heightb(#1)}
\newcommand{\lca}[2]{\lcab(#1,#2)}
\newcommand{\levelancestor}[2]{\levelancestorb(#1,#2)}
\newcommand{\levelnext}[1]{\levelnextb(#1)}
\newcommand{\levelprev}[1]{\levelprevb(#1)}
\newcommand{\levellmost}[2]{\levellmostb(#1,#2)}
\newcommand{\levelrmost}[2]{\levelrmostb(#1,#2)}
\newcommand{\enclose}[1]{\encloseb(#1)}

\newcommand{\sumb}{\mathrm{sum}}
\newcommand{\fwdsearchxb}{\mathrm{fwd\_search}}
\newcommand{\bwdsearchxb}{\mathrm{bwd\_search}}
\newcommand{\fwdsearchb}{\mathrm{fwd\_search}_\geq}
\newcommand{\bwdsearchb}{\mathrm{bwd\_search}_\geq}
\newcommand{\rmqb}{\mathrm{rmq}}
\newcommand{\rmqib}{\mathrm{rmqi}}
\newcommand{\RMQb}{\mathrm{RMQ}}
\newcommand{\RMQib}{\mathrm{RMQi}}
\newcommand{\mincountb}{\mathrm{min\_count}}
\newcommand{\minselectb}{\mathrm{min\_select}}
\newcommand{\minsumb}{\mathrm{min\_sum}}
\newcommand{\cfwdsearchb}{\mathrm{compute\_fwd\_search}}
\newcommand{\cmincountb}{\mathrm{compute\_min\_count}}

\newcommand{\sumf}[4]{\sumb(#1,#2,#3,#4)}
\newcommand{\fwdsearchfx}[4]{\fwdsearchxb(#1,#2,#3,#4)}
\newcommand{\bwdsearchfx}[4]{\bwdsearchxb(#1,#2,#3,#4)}
\newcommand{\fwdsearchf}[4]{\fwdsearchb(#1,#2,#3,#4)}
\newcommand{\bwdsearchf}[4]{\bwdsearchb(#1,#2,#3,#4)}
\newcommand{\rmqf}[4]{\rmqb(#1,#2,#3,#4)}
\newcommand{\rmqif}[4]{\rmqib(#1,#2,#3,#4)}
\newcommand{\RMQf}[4]{\RMQb(#1,#2,#3,#4)}
\newcommand{\RMQif}[4]{\RMQib(#1,#2,#3,#4)}
\newcommand{\mincountf}[4]{\mincountb(#1,#2,#3,#4)}
\newcommand{\minselectf}[5]{\minselectb(#1,#2,#3,#4,#5)}

\newcommand{\fwdsearch}[3]{\fwdsearchb(#1,#2,#3)}
\newcommand{\bwdsearch}[3]{\bwdsearchb(#1,#2,#3)}
\newcommand{\rmq}[3]{\rmqb(#1,#2,#3)}
\newcommand{\RMQ}[3]{\RMQb(#1,#2,#3)}
\newcommand{\mincount}[3]{\mincountb(#1,#2,#3)}
\newcommand{\minsum}[4]{\mincountb(#1,#2,#3,#4)}
\newcommand{\cfwdsearch}[3]{\cfwdsearchb(#1,#2,#3)}
\newcommand{\cmincount}[3]{\cmincountb(#1,#2,#3)}

\newcommand{\arraymax}[1]{M^f_{#1}}
\newcommand{\arraymaxdiff}[1]{D^f_{#1}}
\newcommand{\arraylast}[1]{L^f_{#1}}
\newcommand{\arraylastpi}[1]{L^\pi_{#1}}
\newcommand{\arraymin}[1]{m^\pi_{#1}}
\newcommand{\arraymincount}[1]{N^\pi_{#1}}
\newcommand{\arraymincountcap}[1]{\hat{N}^\pi_{#1}}

\newcommand{\rank}[3]{\mathrm{rank}_{#1}(#2,#3)}
\newcommand{\select}[3]{\mathrm{select}_{#1}(#2,#3)}
\newcommand{\next}[3]{\mathrm{next}_{#1}(#2,#3)}
\newcommand{\prev}[3]{\mathrm{prev}_{#1}(#2,#3)}

\newcommand{\substr}[3]{{#1[#2..#3]}}

\makeatletter
\def\@endtheorem{\endtrivlist}
\makeatother

\begin{document}

\title{Succinct data structure for dynamic trees with faster queries}
\author{Dekel Tsur%
\thanks{Department of Computer Science, Ben-Gurion University of the Negev.
Email: \texttt{dekelts@cs.bgu.ac.il}}}
\date{}
\maketitle

\begin{abstract}
Navarro and Sadakane [TALG 2014] gave a dynamic succinct data structure for
storing an ordinal tree.
The structure supports tree queries in either $O(\log n/\log\log n)$
or $O(\log n)$ time, and insertion or deletion of a single node in $O(\log n)$
time.
In this paper we improve the result of Navarro and Sadakane by reducing the
time complexities of some queries (e.g.\ degree and level\_ancestor) from
$O(\log n)$ to $O(\log n/\log\log n)$.
\end{abstract}

\section{Introduction}
A problem which was extensively studied in recent years is designing a
succinct data structure that stores a tree while supporting queries on the tree,
like finding the parent of a node, or computing the lowest common ancestor
of two nodes.
This problem has been studied both for static trees~\cite{Jacobson89,MunroR01,BenoitDMRRR05,DelprattRR06,GearyRR06,GearyRRR06,GolynskiGGRR07,RamanRS07,HeMS12,JanssonSS12,MunroRRR12,FarzanM14,NavarroS14}
% ChiangLL05
%A partial list of tree queries are given in Table~\ref{tab:queries}.
and dynamic
trees~\cite{MunroRS01,RamanR03,FarzanM11,ArroyueloDS16,GuptaHSV07,NavarroS14}.
% Removed: JoannouR12

%Previous work on dynamic succinct data structures for ordinal trees have
%trade-offs between query time, update time, and the set of supported operations.
For dynamic ordinal trees,
Farzan and Munro~\cite{FarzanM11} gave a data structure with
$O(1)$ query time and $O(1)$ amortized update time.
However, the structure supports only a limited set of queries, and the update
operations are restricted (insertion of a leaf, insertion of a node in the
middle of an edge, deletion of a leaf, and deletion of a node with one child).
A wider set of queries is supported by the data structure of
Gupta et al.~\cite{GuptaHSV07}. This data structure has $O(\log\log n)$ query
time and $O(n^\epsilon)$ amortized update time.
The data structure of Navarro and Sadakane~\cite{NavarroS14} supports a large
set of queries.
See Table~\ref{tab:queries} for some of the supported queries.
The structure supports the following update operations
(1) Insertion of a node $x$ as a child of an existing node $y$.
The insert operation specifies a (possibly empty)
consecutive range of children of $y$ and these nodes become children of
$x$ after the insertion.
(2) Deletion of a node $x$. The children of $x$ become children of the parent of
$x$.
The time complexity of a query is either $O(\log n/\log\log n)$ or
$O(\log n)$ (see Table~\ref{tab:queries}). Moreover, the time complexity of
insert and delete operations is $O(\log n)$.
Additionally, by dropping support for $\degreeb$, $\childrankb$, and
$\childselectb$ queries, the time complexity of insert and delete operations
can be reduced to $O(\log n/\log\log n)$.

In this paper, we improve the result of Navarro and Sadakane by reducing the
time for the queries $\levelancestorb$, $\levelnextb$, $\levelprevb$,
$\levellmostb$, $\levelrmostb$, and $\degreeb$ from $O(\log n)$
to $O(\log n/\log\log n)$. The time complexities of the other operations
are unchanged.
Additionally, by dropping support for $\degreeb$, $\childrankb$, and
$\childselectb$ queries, we obtain a data structure that handles all
queries and update operations in $O(\log n/\log\log n)$ time.

\begin{table}
\caption{Some of the tree queries supported by the data structure of
Navarro and Sadakane~\cite{NavarroS14}.
In the table below $x$ is some node of the tree.
The queries marked by * take $O(\log n)$ time in the structure of Navarro and
Sadakane, and $O(\log n/\log\log n)$ time in our structure.
The queries marked by + take $O(\log n)$ time in both structures, and
unmarked queries take $O(\log n/\log\log n)$ time in both structures.
\label{tab:queries}}
\medskip
\centering
\begin{tabular}{lp{10.5cm}}
\toprule
Query & Description \\
\midrule
$\depth{x}$ & The depth of $x$. \\
$\height{x}$ & The height of $x$. \\
$\numdescendants{x}$ & The number of descendants of $x$. \\
$\parent{x}$ & The parent of $x$. \\
$\lca{x}{y}$ & The lowest common ancestor of $x$ and $y$. \\
$\levelancestor{x}{i}^*$ & The ancestor $y$ of $x$ for which
						$\depth{y} = \depth{x}-i$. \\
$\levelnext{x}^*$ &  The node after $x$ in the BFS order\\
$\levelprev{x}^*$ &  The node before $x$ in the BFS order\\
$\levellmost{x}{d}^*$ & The leftmost node with depth $d$.\\
$\levelrmost{x}{d}^*$ & The rightmost node with depth $d$.\\
$\degree{x}^*$ & The number of children of $x$. \\
$\childrank{x}^+$ & The rank of $x$ among its siblings. \\
$\childselect{x}{i}^+$ & The $i$-th child of $x$. \\
$\firstchild{x}$ & The first child of $x$.\\
$\lastchild{x}$ & The last child of $x$.\\
$\nextsibling{x}$ & The next sibling of $x$.\\
$\prevsibling{x}$ & The previous sibling of $x$.\\
\bottomrule
\end{tabular}
\end{table}

The rest of the paper is organize as follows.
In Section~\ref{sec:partial-sums} we give a dynamic partial sums structure that
will be used later in our data structure.
In Section~\ref{sec:minmax} we give a short description of the data structure
of Navarro and Sadakane. Then, we describe our improved structure in
Sections~\ref{sec:fwd-search} and~\ref{sec:degree}.

\section{Dynamic partial sums}\label{sec:partial-sums}
\newcommand{\partialsum}[2]{\mathrm{sum}(#1,#2)}
\newcommand{\partialsearch}[2]{\mathrm{search}(#1,#2)}
\newcommand{\partialupdate}[3]{\mathrm{update}(#1,#2,#3)}
\newcommand{\partialmerge}[2]{\mathrm{merge}(#1,#2)}
\newcommand{\partialdivide}[3]{\mathrm{divide}(#1,#2,#3)}

In the \emph{dynamic partial sums problem}, the goal is to store
an array $Z$ of integers and support the following queries.
\begin{description}
\item[$\partialsum{Z}{i}$:] Return $\sum_{j=1}^i Z[i]$.
\item[$\partialsearch{Z}{d}$:] Return the minimum $i$ for which
$\partialsum{Z}{i} \geq d$.
\end{description}
Additionally, the following update operations are supported.
\begin{description}
\item[$\partialupdate{Z}{i}{\Delta}$:] Set $Z[i]\gets Z[i]+\Delta$.
\item[$\partialmerge{Z}{i}$:] Replace the entries $Z[i]$ and $Z[i+1]$ by a new
entry that is equal to $Z[i]+Z[i+1]$.
\item[$\partialdivide{Z}{i}{t}$:] Replace the entry $Z[i]$ by the entries
$t$ and $Z[i]-t$.
%It is assumed that $0\leq t\leq Z[i]$.
\end{description}
Note that a partial sums structure also supports access to $Z$ since
$Z[i] = \partialsum{Z}{i}-\partialsum{Z}{i-1}$.

\begin{lemma}\label{lem:partial-sums}
(Bille et al.~\cite{BilleCGSVV18})
There is a dynamic partial sums structure for an array $Z$ containing
$k = O(\log n/\log\log n)$ $O(\log n)$-bit non-negative integers.
The structure uses $O(k\log n)$ bits and supports all queries and update
operations in $O(1)$ time. The $\partialupdate{Z}{i}{\Delta}$ operation
is supported for values of $\Delta$ satisfying $|\Delta| = \log^{O(1)} n$.
\end{lemma}

In the rest of this section we describe structures for storing 
an array $Z$ with negative integers. These structures support only
subsets of the operations defined above.
\begin{lemma}\label{lem:partial-sums-negative-sum}
(Dietz~\cite{Dietz89})
There is a structure for an array $Z$ containing $k = O(\log n/\log\log n)$ 
$O(\log n)$-bit integers.
The structure uses $O(k\log n)$ bits and supports the following
operations in $O(1)$ time:
(1) $\partialsum{Z}{i}$ queries.
(2) $\partialupdate{Z}{i}{\Delta}$ operations,
where $|\Delta|= \log^{O(1)} n$.
\end{lemma}
\begin{corollary}\label{cor:partial-sums-negative-sum}
There is a structure for an array $Y$ containing $k = O(\log n/\log\log n)$ 
$O(\log n)$-bit integers.
The structure uses $O(k\log n)$ bits and supports the following
operations in $O(1)$ time:
(1) Access $Y[i]$.
(2) Add $\Delta$ to the entries of $\substr{Y}{i}{k}$,
where $|\Delta|= \log^{O(1)} n$.
\end{corollary}
\begin{proof}
Define an array $\substr{Z}{1}{k}$ in which $Z[i] = Y[i]-Y[i-1]$ and store
the structure of Lemma~\ref{lem:partial-sums-negative-sum} on $Z$.
%Adding $\Delta$ to the entries of $\substr{Y}{i}{j}$ is done by performing
%$\partialupdate{Z}{i}{\Delta}$ and $\partialupdate{Z}{j+1}{-\Delta}$.
\end{proof}

%We now give an additional result for an array with negative elements.
\begin{lemma}\label{lem:partial-sums-negative}
There is a structure for an array $Z$ containing $k = O(\log n/\log\log n)$ 
$O(\log n)$-bit integers.
The structure uses $O(k\log n)$ bits and supports the following
operations in $O(1)$ time:
(1) $\partialsum{Z}{i}$ queries.
(2) $\partialsearch{Z}{d}$ queries, where $d > 0$.
(3) $\partialupdate{Z}{i}{\Delta}$ operations, where $\Delta \in \{-1,1\}$.
\end{lemma}
%Our structure have the following restrictions:
%(1) In an $\partialsearch{Z}{d}$ query, $d\geq 0$.
%(2) In an $\partialupdate{Z}{i}{\Delta}$ operation, $\Delta \in \{-1,1\}$.
\begin{proof}
Define $\substr{Y}{0}{k+1}$ to be an array in which
$Y[0] = 0$, $Y[k+1] = \infty$, and $Y[i] = \partialsum{Z}{i}$ for
$1\leq i \leq k$.
Let $\substr{I}{0}{k+1}$ be a binary string in which $I[0] = 1$,
and for $i \geq 1$, $I[i] = 1$ if $\max(\substr{Y}{0}{i-1}) < Y[i]$.
Let $Y'$ be an array containing the entries $Y[i]$ for all indices $i\neq 0,k+1$
for which $I[i] = 1$,
and let $Z'$ be an array of size $|Y'|$ in which $Z'[i] = Y'[i]-Y'[i-1]$
($Z'[1]=Y'[1]$). Note that by definition, $Z'[i] \geq 1$ for all $i$.
Our data structure consists of the following structures.
\begin{itemize}
\item The structure of Lemma~\ref{lem:partial-sums-negative-sum} on $Z$.
\item The string $I$.
\item The structure of Lemma~\ref{lem:partial-sums} on $Z'$.
\item An array $\substr{D}{1}{k}$ in which $D[i]=Y[\prev{1}{I}{i}]-Y[i]$,
where $\prev{1}{I}{i}$ is the maximum index $i' \leq i$ such that $I[i']=1$.
%If no such index exists, $\prev{1}{I}{i} = 0$.
\end{itemize}
See Table~\ref{tab:example} for an example.
\begin{table}
\footnotesize
\begin{tabular}{|c|c|c|c|c|c|c|c|c|c|}
%\hline
% & 1 & 2 & 3 & 4 & 5 & 6 & 7 & 8 & 9 & 10\\
\hline
$Z$ & 2 & -2 & -1 & 3 & -1 & 1 & 1 & -3 & 5\\
\hline
$Y$ & 2 &  0 & -1 & 2 & 1  & 2 & 3 & 0  & 5\\
\hline
$I$ & 1 &  0 & 0  & 0 & 0  & 0 & 1 & 0  & 1\\
\hline
$D$ & 0 &  2 & 3  & 0 & 1  & 0 & 0 & 3  & 0\\
\hline
$Z'$& 2 & 1 & 2 & \multicolumn{6}{c}{}\\
\cline{1-4}
$Y'$& 2 & 3 & 5 & \multicolumn{6}{c}{}\\
\cline{1-4}
\end{tabular}
\hfill
\begin{tabular}{|c|c|c|c|c|c|c|c|c|c|}
%\hline
% &  &  & $i$ &  & $l$ &  & $j$ &  &  & \\
%\hline
% & 1 & 2 & 3 & 4 & 5 & 6 & 7 & 8 & 9 & 10\\
\hline
$Z$ & 2 & \textbf{-1} & -1 & 3 & -1 & 1 & 1 & -3 & 5\\
\hline
\textbf{$Y$} & -1 & \textbf{1} & \textbf{0} & \textbf{3} & \textbf{2} & \textbf{3} & \textbf{4} & \textbf{1} & \textbf{6}\\
\hline
$I$ & 1 & 0 & 0 & \textbf{1} & 0 & 0 & 1 & 0 & 1\\
\hline
$D$ & 0 & \textbf{1} & \textbf{2} & 0 & 1 & 0 & 0 & 3 & 0\\
\hline
$Z'$ & 2 & \textbf{1} & 1 & 2 & \multicolumn{5}{c}{}\\
\cline{1-5}
$Y'$ & 2 & 3 & \textbf{4} & \textbf{6} & \multicolumn{5}{c}{}\\
\cline{1-5}
\end{tabular}
\caption{An example showing the arrays of the data structure of
Lemma~\ref{lem:partial-sums-negative-sum}.
The left table gives the values of the array $Z$ and the corresponding
arrays $Y$, $I$, $D$, $Z'$ and $Y'$ (the arrays $Y$ and $I$ are shown without
entries $0$ and $k+1$).
The table on the right shows the array $Z$ and the corresponding arrays
after an $\partialupdate{Z}{2}{1}$
operation.
Changed entries appear in bold.
\label{tab:example}}
\end{table}

To answer a $\partialsearch{Z}{d}$ query,
compute $i = \partialsearch{Z'}{d}$ and return $\select{1}{I}{i}$.
The computation of $\select{1}{I}{i}$ is done in $O(1)$ time
using a lookup table. Therefore, the query is handled in $O(1)$ time.

We next describe how to handle an $\partialupdate{Z}{i}{\Delta}$ operation
(recall that $\Delta \in \{-1,1\}$).
\begin{enumerate}
\item Perform an $\partialupdate{Z}{i}{\Delta}$ operation on the structure of
 Lemma~\ref{lem:partial-sums-negative-sum}.
\item $j\gets \next{1}{I}{i}$ (namely, $j\geq i$ is the minimum index such that
$I[j]=1$).
%, and $j=k+1$ if no such index exists.
\item $i'\gets \rank{1}{I}{j}$.
\item $\partialupdate{Z'}{i'}{\Delta}$.
\item If $i<j$ and $\Delta = 1$:
\begin{enumerate}
\item Let $l$ be the minimum index such that $D[l] = 0$.
If no such index exists, $l=k+1$.
\label{alg:l}
%($l$ is well-defined since $D[j] = 0$).
\item Add $-1$ to the entries of $\substr{D}{i}{l-1}$.\label{alg:D-1}
\item If $l\neq j$, set $I[l]\gets 1$ and perform
$\partialdivide{Z'}{i'}{1}$.
\end{enumerate}
\item If $i<j$ and $\Delta = -1$:
\begin{enumerate}
\item Add $1$ to the entries of $\substr{D}{i}{j-1}$.\label{alg:D+1}
\item If $Z'[i'] = 0$, set $I[j]\gets 0$ and perform $\partialmerge{Z'}{i'}$.
\end{enumerate}
\end{enumerate}

We now show the correctness of the above algorithm.
We will only prove correctness for the case $\Delta = 1$. The proof
for $\Delta = -1$ is similar.
%The proof for an $\partialupdate{Z}{i}{-1}$ operation is similar and thus
%was omitted.

Consider an $\partialupdate{Z}{i}{1}$ operation.
The update operation causes the entries of $\substr{Y}{i}{k}$ to increase by
$1$.
Recall that for an index $p$, $I[p] = 1$ if $\max(\substr{Y}{0}{p-1}) < Y[p]$.
By definition, $\max(\substr{Y}{0}{p-1}) = Y[\prev{1}{I}{p-1}]$.
If $p > j$ then $\prev{1}{I}{p-1}\geq j \geq i$.
Therefore, the update operation causes both $\max(\substr{Y}{0}{p-1})$
and $Y[p]$ to increase by $1$. Therefore, the condition
$\max(\substr{Y}{0}{p-1}) < Y[p]$ is satisfied after the update if and only
if it was satisfied before the update. In other words, the value of
$I[p]$ does not change due to the update operation.
For $p < i$, both $\max(\substr{Y}{0}{p-1})$ and $Y[p]$ do not change,
and thus $I[p]$ does not change.
For the index $p = j$, $I[j]=1$,
and thus $\max(\substr{Y}{0}{j-1}) < Y[j]$ before the update.
The update increases $Y[j]$ by $1$, and either increases by $1$ or does not
change $\max(\substr{Y}{0}{j-1})$. Therefore, $\max(\substr{Y}{0}{j-1}) < Y[j]$
after the update, so $I[j]$ does not change.
If $i=j$ we have shown that $I[p]$ does not change for every index $p$.
Therefore, the algorithm correctly updates the array $I$ in this case.

Suppose now that $i<j$.
%Assume that $l<j$ (we omit the proof for the case $l=j$ which is similar).
For $p\in[i,l-1]$, $\max(\substr{Y}{0}{p-1}) = Y[\prev{1}{I}{p-1}] > Y[p]$
before the update.
Since $\prev{1}{I}{p-1} < i$, we have that $\max(\substr{Y}{0}{p-1})$ does
not change and $Y[p]$ increases by one.
Therefore, $\max(\substr{Y}{0}{p-1}) \geq Y[p]$ after the update.
It follows that $I[p]$ does not change.
Due to the same arguments, $\max(\substr{Y}{0}{l-1}) = Y[l]$ before the
update and $\max(\substr{Y}{0}{l-1}) < Y[l]$ after the update.
Thus, $I[l]$ changes from $0$ to $1$.
Finally, for $p\in[l+1,j-1]$, $\max(\substr{Y}{0}{p-1}) \geq Y[p]$ before
the update.
The update increases both $\max(\substr{Y}{0}{p-1})$ and $Y[p]$ by one
%(since after the update, $\prev{1}{I}{p-1}=l$).
Therefore, $I[p]$ does not change.
We obtained again that the algorithm updates $I$ correctly.
It is easy to verify that the algorithm also updates $D$ correctly.

%Suppose that there is an index $p$ such that $I[p]$ changes due to the
%update, and let $p$ be the minimal such index.
%By the previous paragraph, $p\in[i,j-1]$.
%Before the update $I[p] = 0$, namely $Y[p] \leq \max(\substr{Y}{0}{p-1})$.
%Since the update increases $Y[p]$ by $1$, we
%must have that $Y[p] = \max(\substr{Y}{0}{p-1})$ before the update,
%and $Y[p] = \max(\substr{Y}{0}{p-1})+1$ after the update.
%Now, before the update $Y[p]=Y[i']$, so $D[p] = 0$. It follows that $p=l$.
%For every index $q\in[p+1,j-1]$ we have that both
%$\max(\substr{Y}{0}{q-1})$ and $Y[q]$ increase by $1$ by the update, 
%so $I[q]$ does not change.
%We have shown that the algorithm updates $I$ correctly.

The above algorithm takes $O(k)$ time due to lines~\ref{alg:l}, \ref{alg:D-1},
and~\ref{alg:D+1}.
To reduce the time to $O(1)$ we use the following approach from
Navarro and Sadakane. %~\cite{NavarroS14}.
Instead of storing $D$, the data structure stores an array $\hat{D}$ that has
the following properties:
(1) $\hat{D}[i] = 0$ if and only if $D[i] = 0$.
(2) $0 \leq \hat{D}[i] \leq k$ for all $i$.
Due to the first property, we can use $\hat{D}$ instead of $D$ in
line~\ref{alg:l} above.
Moreover, due to the second property, the space for storing $D$ is
$k\lceil\log (k+1)\rceil = O(\log n)$ bits.
Thus, line~\ref{alg:l} can be performed in $O(1)$ time using a lookup table.

The array $\hat{D}$ is updated as follows.
The structure keeps an index $\alpha$.
If $\Delta = -1$, instead of line~\ref{alg:D+1} above,
first perform $\hat{D}[p] \gets \min(k,\hat{D}[p]+1)$ for all
$i \leq p \leq j-1$.
This takes $O(1)$ time using a lookup table.
Additionally, set $\hat{D}[\alpha] \gets
\partialsum{Z}{\prev{1}{I}{\alpha}}-\partialsum{Z}{\alpha}$
(so $\hat{D}[\alpha]=D[\alpha]$ after this step).
Finally, update $\alpha$ by $\alpha\gets \alpha+1$ if $\alpha < k$ and
$\alpha \gets 1$ otherwise.
Handling the case $\Delta = 1$ is similar.
It is easy to verify that $\hat{D}$ satisfies the two properties above.
\end{proof}

\section{The min-max tree}\label{sec:minmax}

In this section we describe the data structure of
Navarro and Sadakane~\cite{NavarroS14} for dynamic trees.
Let $T$ be an ordinal tree.
The \emph{balanced parentheses string} of $T$ is a string $P$
obtained by performing a DFS traversal on $T$.
When reaching a node for the first time an opening parenthesis
is appended to $P$, and when the traversal leaves a node, a closing parenthesis
is appended to $P$.
We will assume $P$ is a binary string, where the character $1$ encodes an
opening parenthesis and $0$ encodes a closing parenthesis.
We also assume that a node $x$ in $T$ is represented by the index of its
opening parenthesis in $P$.
For example, consider a tree $T$ with $3$ nodes in which the root has $2$
children. The balanced parenthesis string of $T$ is $P=110100$,
and the second child of the root is represented by the index $4$.

For a binary string $P$ and a function $f\colon\{0,1\}\to \{-1,0,1\}$,
the following queries are called \emph{base queries}.
\begin{align*}
\sumf{P}{f}{i}{j} & =\sum_{k=i}^{j}f(P[k])\\
\fwdsearchfx{P}{f}{i}{d} & =\min\{j\geq i:\sumf{P}{f}{i}{j} = d\}\\
\bwdsearchfx{P}{f}{i}{d} & =\max\{j\leq i:\sumf{P}{f}{j}{i} = d\}\\
\rmqf{P}{f}{i}{j} & =\min\{\sumf{P}{f}{1}{k}:i\leq k\leq j\}\\
\rmqif{P}{f}{i}{j} & = \min\{i \leq k \leq j:
	\sumf{P}{f}{1}{k} = \rmqf{P}{f}{i}{j} \} \\
\mincountf{P}{f}{i}{j} & = |\{i\leq k\leq j:
	\sumf{P}{f}{1}{k} = \rmqf{P}{f}{i}{j} \}| \\
\minselectf{P}{f}{i}{j}{d} & = \text{The $d$-th smallest element of }\\
& \phantom{=}\{i\leq k\leq j: \sumf{P}{f}{1}{k} = \rmqf{P}{f}{i}{j} \} \\
\RMQf{P}{f}{i}{j} & =\max\{\sumf{P}{f}{1}{k}:i\leq k\leq j\}\\
\RMQif{P}{f}{i}{j} & = \min\{i \leq k \leq j:
	\sumf{P}{f}{1}{k} = \RMQf{P}{f}{i}{j} \}
\end{align*}
Navarro and Sadakane showed that in order to support queries on the tree $T$,
it suffices to support the following base queries, where $P$ is the
balanced parentheses string of $T$.
\begin{itemize}
\item All base queries on a function $\pi$ defined by
$\pi(1) = 1$ and $\pi(0) = -1$.
\item $\sumb$ and $\fwdsearchxb$ queries on a function $\phi$ defined by
$\phi(1) = 1$ and $\phi(0) = 0$.
\item $\sumb$ and $\fwdsearchxb$ queries on a function $\psi$ defined by
$\psi(1) = 0$ and $\psi(0) = 1$.
\end{itemize}

For example, $\levelancestor{x}{d}=\bwdsearchfx{P}{\pi}{x}{d+1}$.
As noted in Tsur~\cite{Tsur_labeled},
the base queries $\fwdsearchxb$ and $\bwdsearchxb$
can be replaced by the following queries:
\begin{align*}
\fwdsearchf{P}{f}{i}{d} & =\min\{j\geq i:\sumf{P}{f}{i}{j} \geq d\}\\
\bwdsearchf{P}{f}{i}{d} & =\max\{j\leq i:\sumf{P}{f}{j}{i} \geq d\}
\end{align*}
We now need to support the base query $\bwdsearchb$ on the functions $\pi$ and
$\pi' = -\pi$ (namely, $\pi'(1) = -1$ and $\pi'(0) = 1$)
and the base query $\fwdsearchb$ on the functions $\pi$, $\pi'$, $\phi$,
and $\psi$.

To support the base queries, it is convenient to use an equivalent formulation
of these queries.
For an array of integers $A$, let
\[ \fwdsearch{A}{i}{d} = \min\{j \geq i: A[j] \geq d\}. \]
%If there is no index $j\geq i$ for which $A[j] \geq d$, we define
%$\fwdsearch{A}{i}{d} = \infty$.
%We also define $\FLV{A}{d} = \fwdsearch{A}{1}{d}$.
For a binary string $P$, let $f(P)$ be an array of length $|P|$,
where $f(P)[i] = \sumf{P}{f}{1}{i}$.
Then,
\[ \fwdsearchf{P}{f}{i}{d} = \fwdsearch{f(P)}{i}{d+f(P)[i-1]}\]
The other base queries on $f$ can also be rephrased accordingly.

%We note that Navarro and Sadakane used a slightly different definition for
%$\fwdsearchb$ and $\bwdsearchb$, but their technique is easily modified for the
%alternative definitions above.

In order to support the base queries, the string $P$ is partitioned
into blocks of sizes $\Theta(\log^2 n/\log \log n)$.
The blocks are kept in a B-tree, called a \emph{min-max tree},
where each leaf stores one block. 
The degrees of the internal nodes of the min-max tree are
$\Theta(\sqrt{\log n})$,
and therefore the height of the tree is $\Theta(\log n/\log\log n)$.
Each internal node stores \emph{local structures} that are used for
answering the base queries.
A base query is handled by going down from the root of the min-max tree to one
or two leaves of the tree, while performing queries on the local structures
of the internal nodes that are traversed.
%The time of a base query is thus $O(\log n/\log \log n)$ times
%the time spent in each node of the traversal.

The tree queries $\levelancestorb$, $\levelnextb$, $\levelprevb$,
$\levellmostb$, and $\levelrmostb$ are handled by performing a
$\fwdsearch{f(P)}{i}{d}$ or a $\bwdsearch{f(P)}{i}{d}$ query.
These queries take $O(\log n)$ time in the data structure of Navarro and
Sadakane.
In Section~\ref{sec:fwd-search} we will show how to reduce the time of
$\fwdsearchb$ queries to $O(\log n/\log\log n)$
(the handling of $\bwdsearchb$ queries is similar and thus omitted).
In Section~\ref{sec:degree} we will show how to support $\degreeb$ queries
in $O(\log n/\log\log n)$ time.

\section{fwd\_search queries}\label{sec:fwd-search}

In this section we describe how to support $\fwdsearch{f(P)}{i}{d}$ queries
in $O(\log n/\log\log n)$ time.
We first describe how these queries are handled in $O(\log n)$ time in the
structure of Navarro and Sadakane.
For each node $v$ in the min-max tree, let $P_v$ be the substring of $P$
obtained by concatenating the blocks of the descendant leaves of $v$.
Suppose $v$ is an internal node of the min-max tree and the children of $v$
are $v_1,\ldots,v_k$.
We partition $f(P_v)$ into blocks $f(P_v)_1,\ldots,f(P_v)_k$ where the
size of $i$-th block is $|P_{v_i}|$.
Note that $f(P_v)_t[i] = f(P_{v_t})[i]+\delta_t$ for all $i$, where
$\delta_t$ is the last element of $f(P_v)_{t-1}$.

In data structure of Navarro and Sadakane, each internal node $v$
of the min-max tree stores the following local structures.
\begin{itemize}
\item The structure of Lemma~\ref{lem:partial-sums} on an
array $\substr{S_v}{1}{k}$ in which $S_v[i]$ is the size of $P_{v_i}$.
\item A structure supporting $\fwdsearchb$ queries on an array
$\substr{\arraymax{v}}{1}{k}$ in which $\arraymax{v}[i] = \max(f(P_v)_i)$.
\item The structure of Corollary~\ref{cor:partial-sums-negative-sum} on
an array $\substr{\arraylast{v}}{1}{k}$ in which
$\arraylast{v}[i]$ is the last entry of $f(P_v)_{i-1}$.
\end{itemize}
We now give a recursive procedure $\cfwdsearch{v}{i}{d}$ that computes
$\fwdsearch{f(P_v)}{i}{d}$.
\begin{enumerate}
\item If $v$ is a leaf in the min-max tree, compute the answer using a lookup
table and return it.
\item If $i = 1$
\begin{enumerate}
\item $t \gets 0$.
\end{enumerate}
else
\begin{enumerate}
\setcounter{enumii}{1}
\item $t \gets \partialsearch{S_v}{i}$.
%\item $s \gets \partialsum{S_v}{t-1}$.
\item $j'\gets
 \cfwdsearch{v_t}{i-\partialsum{S_v}{t-1}}{d-\arraylast{v}[t]}$.
\label{alg:fwdsearch1}
\item If $j'\neq \infty$, return $j'+ \partialsum{S_v}{t-1}$.
\end{enumerate}
\item $t' \gets \fwdsearch{\arraymax{v}}{t+1}{d}$.
\label{alg:fwdsearch2}
\item If $t' = \infty$ return $\infty$.
\label{alg:if}
\item Return
$\cfwdsearch{v_{t'}}{1}{d-\arraylast{v}[t']} + \partialsum{S_v}{t'-1}$.
\label{alg:fwdsearch3}
\label{alg:FLV}
\end{enumerate}

The time of step~1 is $O(\log n/\log\log n)$.
Navarro and Sadakane showed that $\fwdsearchb$ queries on
$\arraymax{v}$ can be handled in $O(\log k) = O(\log\log n)$ time.
During the computation of $\fwdsearch{f(P)}{i}{d}$, the procedure $\cfwdsearchb$
is called on $O(\log n/\log \log n)$ nodes of the min-max tree.
Therefore, the time for a $\fwdsearch{f(P)}{i}{d}$ query is $O(\log n)$.

When a single character is inserted or deleted from $P$,
the local structures of $O(\log n/\log\log n)$ nodes in the min-max tree
are updated:
If $v$ is the leaf whose block contains the inserted or deleted character,
only the local structures of the ancestors of $v$ are updated, assuming no split
or merge operations were used to rebalance the min-max tree.
The cost of splitting or merging min-max nodes can be ignored if an appropriate
B-tree balancing algorithm is used (Navarro and Sadakane used the
balancing algorithm of Fleischer~\cite{Fleischer96}, but other balancing
algorithms can be used, e.g.\ the algorithm of Willard~\cite{Willard00}).
Each update takes $O(1)$ time, and therefore updating all local
structures for a single character update on $P$ takes $O(\log n/\log\log n)$
time.
%The actual cost of a single character update is $O(\log n)$ since
%updating the local structures for $\mincountb$ queries takes
%$O(\log k) = O(\log\log n)$ time per min-max tree node.
An insertion or deletion of a node from $T$ consists of insertion or deletion
of two characters from $P$. Therefore, the time to update the local structures
is $O(\log n/\log\log n)$.

We now describe how to support $\fwdsearchb$ queries in $O(\log n/\log\log n)$
time.
In addition to the local structures described above, we also store the following
local structures in each internal node $v$ of the min-max tree.
\begin{itemize}
\item An RMQ structure on $\arraymax{v}$. Like in the structure of 
Navarro and Sadakane,
this RMQ structure consists of the balanced parentheses string of the
max-Cartesian tree of $\arraymax{v}$.
\item The structure of Lemma~\ref{lem:partial-sums-negative}
on an array $\substr{\arraymaxdiff{v}}{1}{k}$ in which
$\arraymaxdiff{v}[i] = \arraymax{v}[i]-\arraymax{v}[i-1]$.
\end{itemize}
%We node that the structure of Navarro and Sadakane already uses an RMQ
%structure on $\arraymax{\Pfv{v}}$ in order to support rmq queries on $f(P)$.

The procedure $\cfwdsearchb$ is changed by replacing lines~\ref{alg:fwdsearch2}
and~\ref{alg:if} with the following equivalent lines:
\begin{enumerate}
\setcounter{enumi}{4}
\item If $\RMQ{\arraymax{v}}{t+1}{k} < d$ return $\infty$.\label{alg:RMQ}
\item If $i = 1$ then
\begin{enumerate}
\item  $t'\gets \partialsearch{\arraymaxdiff{v}}{d}$.\label{alg:partial-search}
\end{enumerate}
else
\begin{enumerate}
\setcounter{enumii}{1}
\item $t' \gets \fwdsearch{\arraymax{v}}{t+1}{d}$.\label{alg:fwdsearch}
\end{enumerate}
\end{enumerate}

Consider the computation of $i^*=\fwdsearch{f(P)}{i}{d}$ using 
procedure $\cfwdsearchb$.
Let $w$ (resp., $w^*$) be the leaf in the min-max tree whose block contains
$P[i]$ (resp., $P[i^*]$).
Let $u^*_1,u^*_2,\ldots,u^*_h = w^*$ be the nodes on the path from the root of
the min-max tree to $w^*$, and let $u^*_s$ be the lowest common ancestor of $w$
and $w^*$.
Let $u_s = u^*_s, u_{s+1}, \ldots, u_h = w$ be the nodes on the path
from $u^*_s$ to $w$.
The computation of $\fwdsearch{f(P)}{i}{d}$ makes the following
calls to $\cfwdsearchb$. First, the procedure is called on
$u^*_1,u^*_2,\ldots,u^*_s$.
Then, the procedure is called on $u_{s+1},\ldots,u_h$.
%In these calls, except the last, the procedure terminates at line~\ref{alg:RMQ}.
Finally, the procedure is called on $u^*_{s+1},\ldots,u^*_h$.
Note that line~\ref{alg:fwdsearch} is executed only when the procedure is called
on $u^*_s$.
Therefore, line~\ref{alg:fwdsearch} contributes $O(\log k)=O(\log\log n)$ time
to the total time of the computation.
The rest of the recursive calls, except the two calls on $u_h$ and $u^*_h$,
take $O(1)$ time each.
Therefore, the total time is $O(\log n/\log\log n)$.

Recall that the structure of Lemma~\ref{lem:partial-sums-negative} supports
$\partialsearch{\cdot}{d}$ queries only for $d> 0$.
We therefore need to show that $d$ is non-negative in
line~\ref{alg:partial-search}.
Note that this line is executed only when procedure $\cfwdsearchb$ is called
on $u^*_{s+1},\ldots,u^*_{h-1}$.
First, in every tree query that is answered by an
$\fwdsearchf{P}{f}{i}{d}$ query, the parameter $d$ is non-negative.
The $\fwdsearchf{P}{f}{i}{d}$ query is handled by answering an
$\fwdsearch{f(P)}{i}{d'}$ query, where $d'= d + f(P)[i-1] =
 d + \sumf{P}{f}{1}{i-1}$.

Let $l_j,r_j$ be the indices such that $P_{u^*_j} = \substr{P}{l_j}{r_j}$
(note that $i^*\in[l_j,r_j]$ for all $j$).
When procedure $\cfwdsearchb$ is called on $u^*_j$,
the value of $d'$ is decreased by $\sumf{P}{f}{l_{j-1}}{l_j-1}$.
It follows that when the procedure is called on $u^*_j$ for $j>s$,
the value of $d'$ is $d' = d - \sumf{P}{f}{i}{l_{j-1}-1}$.
Therefore, $d' > 0$ otherwise $\fwdsearchf{P}{f}{i}{d} \leq l_{j-1}-1 < i^*$
which contradicts the definition of $i^*$.

\paragraph{Updating the structures}
We now show how to update the additional local structures
when a character is inserted or deleted from $P$.
Navarro and Sadakane showed that the RMQ structure on $\arraymax{v}$
can be updated in $O(1)$ time using a lookup table.
A single character update on $P$ either does not change $\arraymax{v}$,
increases the entries of $\substr{\arraymax{v}}{i}{k}$ by $1$ for some $i$, or
decreases the entries of $\substr{\arraymax{v}}{i}{k}$ by $1$ for some $i$.
Therefore, either the array $\arraymaxdiff{v}$ does not change,
or a single entry of $\arraymaxdiff{v}$ is either increased by $1$ or decreased
by $1$.
Thus, an $\partialupdate{\arraymaxdiff{v}}{i}{\pm1}$ operation updates the
structure on $\arraymaxdiff{v}$ in $O(1)$ time.

\section{degree queries}\label{sec:degree}
In this section we show how to handle $\degree{x}$ queries in
$O(\log n/\log\log n)$ time.
We first describe the handling of these queries in the structure of
Navarro and Sadakane.

%We now describe how to answer $\degree{x}$ queries on light nodes.
To compute $\degree{x}$, we use the equality
$\degree{x} = \mincountf{P}{\pi}{x+1}{\linebreak[0]\enclose{x}-1}$,
where $\enclose{x}$ is the index in $P$ of the closing parenthesis of $x$.
As in Section~\ref{sec:minmax}, we will use an equivalent formulation
of $\mincountb$.
For arrays of integers $A$ and $B$ define
\begin{align*}
\mincount{A}{i}{j} & = | \{i \leq k \leq j: A[k] = \min(\substr{A}{i}{j}) \}\\
\minsum{A}{B}{i}{j} & = \sum_{i \leq k\leq j: A[k] = \min(\substr{A}{i}{j})}
B[k]
\end{align*}
We have that $\mincountf{P}{\pi}{i}{j} = \mincount{\pi(P)}{i}{j}$.
In the following we show a structure for computing $\mincount{\pi(P)}{i}{j}$.
%assuming that $\mincount{\pi(P)}{i}{j} < D$.

Consider some internal node $v$ in the min-max tree, and let $k$ be the number
of children of $v$.
Recall that the string $\pi(P_v)$ is partitioned into $k$ blocks
$\pi(P_v)_1,\ldots,\pi(P_v)_k$.
%Let $\substr{\arraymincount{v}}{1}{k}$ be an array in which
%$\arraymincount{v}[i] = \mincount{\pi(P_v)_i}{1}{S_v[i]}$.
The structure of Navarro and Sadakane stores in $v$ the following local
structures.
\begin{itemize}
\item The structure of Corollary~\ref{cor:partial-sums-negative-sum} on
an array $\substr{\arraymin{v}}{1}{k}$ in which
$\arraymin{v}[i] = \min(\pi(P_v)_i)$.
\item An rmq structure on $\arraymin{v}$ (as before, this structure
consists of the balanced parentheses string of the min-Cartesian tree of
$\arraymin{v}$).
\item An array $\substr{\arraymincount{v}}{1}{k}$ in which
$\arraymincount{v}[i] = \mincount{\pi(P_v)_i}{1}{S_v[i]}$.
%\item An array $\substr{\arraymincountcap{v}}{1}{k}$ in which
%$\arraymincountcap{v}[i] = \min(D, \arraymincount{v}[i])$.
\item A structure for answering
$\minsumb$ queries on $\arraymin{v},\arraymincount{v}$.
% (this structure will be used only for updating the array
%$\arraymincountcap{v}$).
\end{itemize}

The following procedure $\cmincount{v}{i}{j}$ returns the pair
\[ \mincount{\pi(P_v)}{i}{j},\rmq{\pi(P_v)}{i}{j}.\]
%if $\mincount{\pi(P_v)}{i}{j} < D$.
\begin{enumerate}
\item If $v$ is a leaf in the min-max tree, compute the answer using a lookup
table and return it.
\item $t \gets \partialsearch{S_v}{i}$ and $t' \gets \partialsearch{S_v}{j}$.
\item $s \gets \partialsum{S_v}{t-1}$ and $s' \gets \partialsum{S_v}{t'-1}$.
\item If $t=t'$:
\begin{enumerate}
\item $N,m \gets \cmincount{v_t}{i-s}{j-s}$.
\item Return $N,m+\arraylastpi{v}[t]$.
\end{enumerate}
\item If $i-s>1$:
\begin{enumerate}
\item $N_1,m_1\gets \cmincount{v_t}{i-s}{S_v[t]}$.
\item $m_1 \gets m_1+\arraylastpi{v}[t]$.
\end{enumerate}
else $m_1\gets \infty$ and $t\gets t-1$.
\item If $j-s' < S_v[t']$:
\begin{enumerate}
\item $N_3,m_3\gets \cmincount{v_{t'}}{1}{j-s'}$.
\item $m_3 \gets m_3+\arraylastpi{v}[t']$.
\end{enumerate}
else $m_3\gets \infty$ and $t'\gets t'+1$.
\item If $t+1 \leq t'-1$:
\begin{enumerate}
\item $m_2 \gets \rmq{\arraymin{v}}{t+1}{t'-1}$.
%\item $N_2 \gets \minsum{\arraymin{v}}{\arraymincountcap{v}}{t+1}{t'-1}$.
\item $N_2 \gets \minsum{\arraymin{v}}{\arraymincount{v}}{t+1}{t'-1}$.
\label{alg:minsum}
\end{enumerate}
else $m_2\gets \infty$.
\item $m \gets \min(m_1,m_2,m_3)$.
\item $N \gets \sum_{l \leq 3:m_l=m} N_l$.
\item Return $N,m$.
\end{enumerate}

The time complexity of line~\ref{alg:minsum} is $O(\log k) = O(\log\log n)$
and therefore the time complexity of a $\minsumb$ query is $O(\log n)$.

We say that a node $x$ of $T$ is \emph{heavy} if it has at least
$D = \lceil\log n\rceil^2$ children.
Our approach for handling $\degree{x}$ queries in $O(\log n/\log\log n)$ time
is to handle differently heavy nodes and light nodes.
In order to handle queries on heavy nodes the data structure stores
the following structures.

\begin{itemize}
\item A rank-select structure on a binary string $\substr{B}{1}{2n}$
in which $B[x] = 1$ if $P[x]$ is an opening parenthesis and $x$ is a heavy node.
\item An array $C$ containing $\degree{x}$ for every $x$ such that $B[x] =1$,
sorted by increasing order of $x$.
\end{itemize}
For both $B$ and $C$ we use dynamic succinct structures from
Navarro and Sadakane~\cite{NavarroS14}.
These structure have $O(\log n/\log\log n)$ query and update time.
Therefore, checking whether a node is heavy, and computing $\degree{x}$ for a
heavy node takes $O(\log n/\log\log n)$ time.
To bound the space for $B$ and $C$,
we use the fact that there are at most $n/D$ heavy nodes.
Therefore, the space for the rank-select structure on $B$ is
$nH_{0}(B)+o(n) = O((n/D)\log\frac{2n}{n/D}) = o(n)$ bits,
and the space for the array $C$ is
$(1+o(1))|C|\log n\leq (1+o(1))n/D\cdot\log n = o(n)$ bits.

For a light node $x$, we compute $\degree{x}$ using $\mincountb$ query.
In addition to the local structures described above,
each internal node $v$ in the min-max tree stores an array
$\substr{\arraymincountcap{v}}{1}{k}$ in which
$\arraymincountcap{v}[i] = \min(D, \arraymincount{v}[i])$.
Recall that $\degree{x} = \mincount{\pi(P)}{x+1}{\enclose{x}-1}$,
and the latter expression can be computed by procedure $\cmincountb$.
Since $\degree{x} < D$, we can replace line~\ref{alg:minsum} in procedure
$\cmincountb$ by
$N_2 \gets \minsum{\arraymin{v}}{\arraymincountcap{v}}{t+1}{t'-1}$.
This line can be performed in constant time as follows.
Using the balanced parenthesis string of the Cartesian tree of $\arraymin{v}$
and a lookup table,
obtain in constant time a binary string $\substr{X}{1}{k}$ such that
$X[p]=1$ if $\arraymin{v}[p] = \rmq{\arraymin{v}}{t+1}{t'-1}$.
Since the space for storing the array $\arraymincountcap{v}$ is
$k\log D = o(\log n)$ bits and the space for storing $X$ is $k = o(\log n)$
bits, a lookup table is used to compute in constant time the sum of
$\arraymincountcap{v}[p]$ for every $p$ such that $X[p]=1$.

\paragraph{Updating the structures}
We first show how to update the local structures in the nodes of the min-max
tree.
Consider some internal node $v$ in the min-max tree, and let
$v_1,\ldots,v_k$ be its children.
The structure of Corollary~\ref{cor:partial-sums-negative-sum} and the rmq
structure on $\arraymin{v}$ can be updated in $O(1)$ time
(see Section~\ref{sec:fwd-search}).
An insertion or deletion of a character from $P$ can only change one entry
of $\arraymincount{v}$, namely the entry
$\arraymincount{v}[i]$ where $i$ is the index such that the changed
character belongs to $P_{v_i}$.
To see why this is true, note that for $j < i$, all values in
$\pi(P_v)_j$ do not change due to the character update.
Therefore, $\arraymincount{v}[j]$ does not change.
Additionally, for $j > i$, either all values in $\pi(P_v)_j$ are increased by
$1$ or all these values are decreased by $1$, so again,
$\arraymincount{v}[j]$ does not change.
Therefore, to update $\arraymincount{v}$ and $\arraymincountcap{v}$, we only
need to compute the value of $\arraymincount{v}[i]$. This can be done in
$O(\log k)=O(\log\log n)$ time by performing
a $\minsum{\arraymin{v_i}}{\arraymincount{v_i}}{1}{S_v[i]}$ query
(recall that the min-max tree nodes have local structures for $\minsumb$
queries).

We next show how to update the structures on $B$ and $C$.
When a node $x$ is inserted to the tree, perform two character insertions
on $B$, and if $x$ is heavy insert its degree to $C$.
Additionally, compute the degree of the parent $y$ of $x$.
%before the insertion.
If the insertion of $x$ changes $y$ from light to heavy or from
heavy to light, update $B$ and $C$ accordingly.
Therefore, the insertion of $x$ causes $O(1)$ changes on $B$ and $C$ which
are performed in $O(\log n/\log \log n)$ time.

We also need to handle the case when insertion or deletions of nodes
causes the value of $\lceil \log n\rceil$ to change.
Since our definition of a heavy node depends on $\lceil \log n\rceil$,
this means that a single node insertion or deletion can cause $\Theta(n/\log n)$
nodes to change their heavy/light status, requiring a $\Omega(n/\log n)$ time
to update $B$ and $C$.
The structure of Navarro and Sadakane already has a mechanism to handle changes
to $\lceil \log n\rceil$ since the sizes of the blocks of $P$ and the sizes
of the lookup tables used by the structure depend on $\lceil \log n\rceil$.
This mechanism works as follows.
The string $P$ is partitioned into three parts $P = P_0 P_1 P_2$.
A separate min-max tree is built on each part $P_i$.
The tree for $P_1$ uses the current value of $\lceil \log n\rceil$, while
the trees for $P_0$ and $P_2$ use $\lceil \log n\rceil-1$ and
$\lceil \log n\rceil+1$, respectively.
When a node is added to or deleted from $T$, the structure changes the partition
of $P$ by moving $O(1)$ characters between parts, and updating
the min-max trees after each movement.
%Using this mechanism
We change the definition of $B$ as follows.
For an index $x$, $B[x] = 1$ if $P[x]$ is an opening parenthesis and
$\degree{x} \leq (\lceil \log n\rceil-2+i)^2$,
where $i$ is the index such that $P[x]$ is in $P_i$.
When the partition of $P$ changes, for every index $x$ such that
$P[x]$ changes its part, we need to compute $\degree{x}$ and then update $B[x]$.
If the value of $B[x]$ changes, a corresponding update on the array $C$ is
performed.
Using this approach, a single node insertion or deletion causes only $O(1)$
changes to $B$ and $C$ (recall that only $O(1)$ characters in $P$ move between
parts).
Additionally, at all times, if $x$ is heavy then $B[x] = 1$.
Thus, the query algorithm remains correct.

\bibliographystyle{plain}
\bibliography{ds,dekel}

\end{document}